\documentclass[conference]{IEEEtran}

\usepackage{url}
\usepackage{ifthen}
\usepackage{cite}
\usepackage{romannum}
\usepackage[cmex10]{amsmath} 
\usepackage{mathtools}
\usepackage{amsthm, amsmath, amssymb}
\usepackage{amsfonts}
\usepackage{subcaption}
 \usepackage{algorithmic}
\usepackage{algorithm}
\usepackage{float}
\newtheorem{theorem}{Theorem}

\begin{document}
\title{Information Bottleneck Methods for Distributed Learning} 
\author{\IEEEauthorblockN{Parinaz Farajiparvar\IEEEauthorrefmark{1}, Ahmad Beirami\IEEEauthorrefmark{2}, and Matthew Nokleby\IEEEauthorrefmark{1}}
\IEEEauthorblockA{\IEEEauthorrefmark{1}Department of Electrical and Computer Engineering, Wayne State University, Detroit, MI\\
\IEEEauthorrefmark{2}Research Laboratory of Electronics, MIT, Cambridge, MA \\
Email: \{parinaz.farajiparvar, matthew.nokleby\}@wayne.edu, berami@mit.edu}}
\maketitle
\begin{abstract}
We study a distributed learning problem in which Alice sends a compressed distillation of a set of training data to Bob, who uses the distilled version to best solve an associated learning problem. We formalize this as a rate-distortion problem in which the training set is the source and Bob's {\em cross-entropy loss} is the distortion measure. We consider this problem for unsupervised learning for batch and sequential data. In the batch data, this problem is equivalent to the {\em information bottleneck} (IB), and we show that reduced-complexity versions of standard IB methods solve the associated rate-distortion problem. For the streaming data, we present a new algorithm, which may be of independent interest, that solves the rate-distortion problem for Gaussian sources. Furthermore, to improve the results of the iterative algorithm for sequential data we introduce a two-pass version of this algorithm. Finally, we show the dependency of the rate on the number of samples $k$ required for Gaussian sources to ensure cross-entropy loss that scales optimally with the growth of the training set.
\end{abstract}


\section{Introduction}
We consider a distributed learning problem in which Alice obtains a training sequence of $k$ i.i.d. samples drawn from a distribution that belongs to a parametric family. Alice wants to distill the training set to a set of features $T$, which she communicates to Bob using no more than $R$ bits. Bob uses $T$ to estimate the data distribution associated with the learning problem. We measure the quality of Bob's distribution according to the {\em cross entropy} loss, which is ubiquitous in machine learning \cite{deng2006cross} and closely related to the KL divergence between the learned and true distributions.

This setting induces a rate-distortion problem: For a given bit budget $R$, what is the encoding $T$ of Alice's training set $X^k$ that minimizes Bob's cross-entropy loss. We consider this problem for both batch and sequential data, and we show that the ideal strategy is to solve a version of the information bottleneck (IB) problem for an appropriate sufficient statistic of $X^k$ \cite{tishby2000information}.

In this work, the associated rate-distortion problem is equivalent to a special case of the information bottleneck \cite{tishby2000information}, in which one wishes to find a compressed representation, $T$, of an observed random variable $X$ that is maximally ``relevant'' to a correlated random variable $Y$, as measured by $I(Y;T)$. IB has been applied to clustering and feature extraction \cite{slonim2000document, slonim2005information}, and Tishby recently proposed an explanation for the success of deep learning in terms of IB \cite{tishby2015deep}. Extensions of IB to distributed \cite{aguerri2017distributed}, interactive \cite{vera2016collaborative}, and multi-layer \cite{yang2017multi} multi-terminal settings have recently been considered. Furthermore, IB was shown to solve distributed learning problems with {\em privacy} constraints \cite{moraffah2017privacy}.

In Section \ref{sect:batch_results}, we show that tailored implementations of IB algorithms proposed in \cite{tishby2000information,chechik2005information} can be used to compute the rate-distortion for discrete and Gaussian data sources. This implementations exploit the fact that using the sufficient statistic reduces the computational/storage complexity of the IB algorithm from exponential in $k$ to polynomial in $k$ and from $kd$-dimensional matrix to $d$-dimensional matrix in discrete and Gaussian sources, respectively. Furthermore, we consider the relationship between the number of samples $k$ and the required compression rate $R$. Indeed, in the {\em data-limited regime} in which $k$ is small, a higher $R$ does not impact the distortion significantly.

In Section \ref{sec:sequential_results}, we consider the encoding of sequential data, where the figure of merit is the total cross-entropy regret. Explicit minimization of the regret turns out to be challenging, so we propose a ``greedy'' on-line method which gives a tractable approach to encoding Gaussian data. In this method, the agent chooses the encoding considers the cross-entropy regret only at the current time instance, and it is provably suboptimum. To improve the regret performance, we also propose a ``two-pass'' solution which includes a backwards pass in which the feature encoding is improved by considering the impact on future regret.

Finally, in Section \ref{sect:conclusion} we draw conclusions and suggest areas for future work.

\section{Problem Statement}\label{sect:problem.setting}
We consider the unsupervised learning problem for both {\em batch} and {\em sequential data}, in which data is distributed according to $p(x|\theta)$ and the objective is to minimize the cross entropy of the learned distribution with $p(x|t)$.

\subsection{Batch data}
Let $(X,\theta)\in(\mathcal{X},\Theta)$ be (discrete or continuous) random variables with joint distribution $p(x|\theta)p(\theta)$. The conditional distribution $p(x | \theta)$ represents a parametric family of distributions on $X$, and $p(\theta)$ represents a (known) prior over the family. Alice does not observe $\theta$ directly, but instead observes a set of $k$ i.i.d. samples $X^k := (X_1,\dots,X_k)$, with $X_i \sim p(x|\theta)$. Alice constructs a distilled representation $T$ of this training set and transmits it to Bob (Figure~\ref{fig:batch_model}). Bob uses $T$ to construct the distribution $p(x|T)$, which approximates both $p(x|X^k)$---the best distribution that can be learned from $X^k$---and $p(x|\theta)$---the true distribution. This gives rise to the Markov chain $X-\theta-X^k-T$, where we emphasize that $X^k$ is the training set, and $X$ is a hypothetical test point conditionally independent of $X^k$ given $\theta$.
\begin{figure}[htb]
\centering
  \includegraphics[width=7cm]{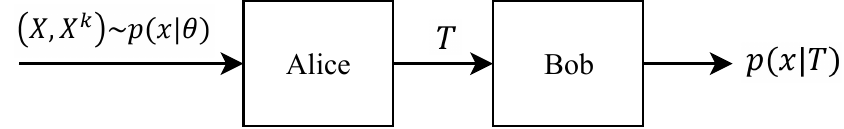}
\caption {Batch data transmission model.}
\label{fig:batch_model}
\end{figure}
Here, we suppose that an encoder has direct access to $\theta$ and wishes to construct a compact representation $T$ to be stored or transmitted to another agent. The representation $T$ is used to construct an approximation $p(x|T)$ of the parametric distribution, which can be used to predict hypothetical test points $X \sim p(x|\theta)$. This gives rise to the Markov chain $X-\theta-T$.

Given a budget of $R$ bits, Alice's objective is to choose $T$ to minimize the {\em cross-entropy loss}, defined as
\begin{equation*}
	H(p(x | \theta) || p(x|T)) = -\mathbb{E}_{p(x | \theta)}[\log p(x|T)].
\end{equation*}
The cross-entropy loss is ubiquitous in machine learning; common practice in deep learning, for example, is to choose model parameters that minimize the empirical cross entropy over the training set \cite{Goodfellow.etal.2016}. The cross-entropy differs from the KL divergence by a constant, and thus measures the distance between the true distribution and $p(x|T)$.\par

Given a stochastic mapping $p(t|x^k)$, the expectation over the cross-entropy loss is $E_{\theta,X^k,T}[H(p(x | \theta) || p(x|T))] = H(X|T)$.\footnote{Although the random variables considered here need not be discrete, in general we use capital $H$ to denote the standard entropy, with the understanding that the differential entropy $h(\cdot)$ is intended when random variables are continuous.} Regarding $I(X^K;T)$ as the average number of bits
required to describe $T$, we define the distortion-rate function as the minimum average cross entropy loss that we achieve when the bit budget is less than $R$:
\begin{equation*}
	D^{k}_\text{B}(R) := \min_{p(t | x^k): I(X^k;T) \leq R} H(X|T),
\end{equation*}
for $H(X|\theta) \leq H(X|X^k) \leq H(X|T) \leq H(X)$ being the range of possible distortion values.
Because $I(X;T) = H(X) - H(X|T)$, finding the distortion-rate function is equivalent to solving the information bottleneck for the Markov chain $X-\theta-T$, i.e. minimizing the mutual information $I(X^k;T)$ subject to a constraint on $I(X;T)$. Indeed, the simple prediction problem can be solved using existing IB techniques for discrete \cite{tishby2000information} or Gaussian \cite{chechik2005information} sources.
\begin{equation*}
  \min_{p(t|x^k)} \mathcal{L}=I(X^k;T)-\beta I(X;T).
\end{equation*} 
However, the dependence of $X$ and $X^k$ through $\theta$ introduces a structure that one can exploit in computing $D^{k}_\text{B}(R)$ and finding the optimum $p(t | x^k)$. For example, a straightforward use of the iterative IB algorithm from \cite{tishby2000information} requires iteration over all $|\mathcal{X}|^k$ possible training sets, which is unmanageable in practice; in Section \ref{sect:batch_results} we show how to reduce the computational and storage burden. Further, when $X$ and $\theta$ are jointly Gaussian, we specialize the results from \cite{chechik2005information} to derive a simple, closed-form expression for $D^{k}_\text{B}(R)$.

The trade-off described by $D^{k}_\text{B}(R)$ improves with larger $k$. The rate-distortion curve may be poor for small $k$, while $\lim_{k \to \infty} D^{k}_\text{B}(R)$ approaches the special case in which Alice has direct access to $\theta$. If $k$ is small, there may be little point in using many bits to describe $X^k$. 

In figure \ref{fig:R-Dcurve} we show upper and lower bounds on the distortion and rate. Also, we provide upper and lower bounds on $R_\text{B}^{k}(D)$
	\begin{align*}
		H(X) - D \leq  \ R_\text{B}^{k}(D) \leq H(\theta)+\log|\mathcal{X}|-D 
\end{align*}
\begin{figure}[htb]
\centering
  \includegraphics[width=7cm]{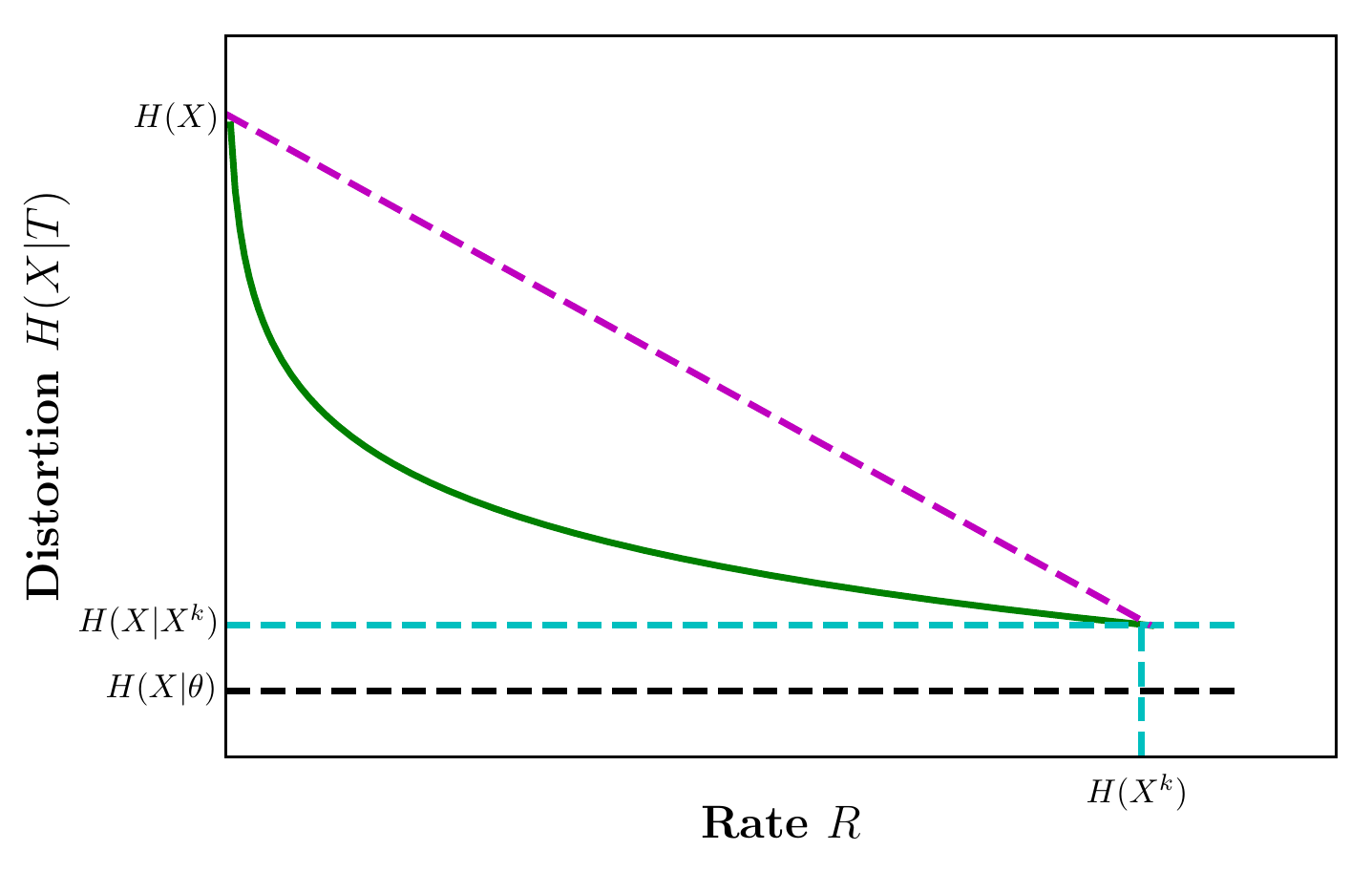}
\caption {Rate-Distortion curve, which shows the outer bounds on the rate and distortion.}
\label{fig:R-Dcurve}
\end{figure}

\subsection{Sequential data}
Next, we consider {\em sequential data}, where again $(X,\theta)\in(\mathcal{X},\Theta)$ be (continuous) random variables with joint distribution $p(x|\theta)p(\theta)$. In this case, instead of observing a set of $k$ i.i.d. samples, Alice observes the data one-by-one in each round of the sequential data transmission, where samples are drawn from $p(x|\theta)$. After each round $l$, Alice constructs a distilled representation $T_l$ of the training set that she observes up to $l$-th round where $1\leq l\leq k$ and transmits it to Bob (Figure~\ref{fig:stream_model}). Then, Bob uses $T^l$ to construct the distribution $p(x|T^l)$, where $T^l := \{T_1,T_2,\dots,T_l\}$.
\begin{figure}[htb]
\centering
  \includegraphics[width=7.5cm]{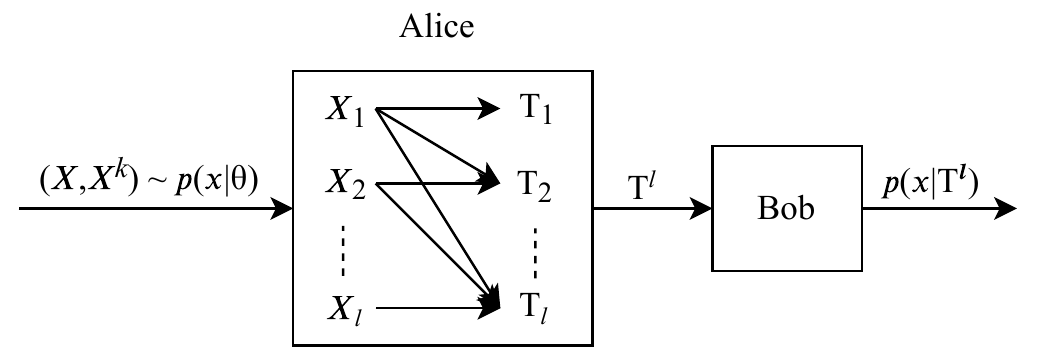}
\caption {Sequential data transmission model.}
\label{fig:stream_model}
\end{figure}
In this set up, the Markov chain for $k$ i.i.d. samples is $X-\theta-X^l-T^k$ and the cross entropy loss at time $l$ is defined as
\begin{equation*}
	H(p(x | \theta) || p(x|T^l)) = -\mathbb{E}_{p(x | \theta)}[\log p(x|T^l)].
\end{equation*}

We introduce three approaches to solve this problem:
\subsubsection{Comprehensive solution}
In this solution, we consider the global design of the features $T_l$ while ensuring that $T_l$ respects causality. In this case, $I(X^l;T_l|T^{-l})$ represents the number of bits required to describe $T^l$ given that the features $T^{-l}$ have already been sent. We take the overall distortion to be the sum regret, or the suffered cross-entropy loss summed over all $k$ samples. Then, distortion-rate function is defined as the minimum distortion that we achieve when the bit budget of each round $1\leq l\leq k$ is less than $R_l$:
\begin{equation*}
\begin{split}
	D^{k}_\text{SC}(R) := \min_{p(t_1|x^1),\dots,p(t_l|x^l):I(X^l;T_l|T^{-l})\leq R_l} \sum_{l=1}^k H(X|T^l).
\end{split}
\end{equation*}
Based on $I(X;T^l)=H(X)-H(X|T^l)$, we can translate the distortion-rate function to the following variational problem: 
\begin{equation*}
\begin{split}
\max_{{p(t_1|x^1),\dots,p(t_l|x^l)}}\!\!\!\!\!\!\!\!\!\!\!\mathcal{L} = &\sum_{l=1}^k I(X;T^l)-\alpha_l I(X^l;T_l|T^{-l})
\end{split}
\end{equation*}
where $T^{-l}:=\{T_1,T_2,\dots,T_{l-1}\}$, $I(X^l;T_l|T^{-l})$ determines the rate in each round, and $\alpha_1,\dots,\alpha_k$ shows the trade-off between the average distortion and the rate of each round. Unfortunately, this problem results in a challenging joint optimization problem over the features $T_l$, and even for Gaussian data, finding a closed-form solution is challenging.

\subsubsection{Online solution}

To find a tractable solution, we present an on-line approach. Instead of find the global solution to the encodings $T_l$, we optimize each term in the objective function one-by-one, without regard for future terms. Therefore, in the $l$th round, we have the distortion-rate function

\begin{equation}
	D^{l}_\text{SO}(R) := \min_{p(t_l | x^l):I(T_l;X^l|T^{-l})\leq R_l} H(X|T^l),
\end{equation}
where the stochastic mapping $p(t_l|x^l)$ gives a ``soft'' description of $T_l$ in each round.\par
In this setup, we define $I(X;T_l|T_{-l})=H(X|T_{-l})-H(X|T^l)$ as the average distortion. This translates the distortion-rate problem to the following minimization problem:
\begin{equation*}
\min_{p(t_l|x^l)}\mathcal{L}=I(X^l;T_l|T^{-l})-\beta_l I(X;T_l|T^{-l}),
\end{equation*} 
where $T_l$ denotes the compressed representation of an observed random variable $X^l$. This formulation is not equivalent to the IB. 

Instead, the equivalent problem is to minimizing the mutual information $I(X^l;T_l|T_{-l})$ given a constraint on the {\em conditional} mutual information $I(X;T_l|T_{-l})$. Consequently, the standard IB algorithms can not be applied here. In Section \ref{sec:sequential_results} we develop a new iterative algorithm for computing $D_\text{SO}^{l}(R)$ based on the sufficient statistic of the Gaussian distribution. 
\subsubsection{Two-path Solution}
In the on-line algorithm presented above, the encoding $T_l$ is chosen supposing that the encoding function for previous features is fixed, and without regard for future features. This results in a strictly suboptimum solution. To improve the solution, we develop a two-pass solution, which adds a backwards pass to the algorithm, taking the future encoding designs as fixed and without regard for {\em previous} feature encodings. After carrying out the on-line algorithm above, we optimize the following loss function for the backward path:
\begin{align*}
\min_{p(t_{l-1}|x^{l-1})}\mathcal{L}\!=I(\overline{X}_{l-1};&T_{l-1}|T^{-(l-1)}\!,T_l)\\
&-\beta I(X;T\!_{l-1}|T^{-(l-1)},T_l)
\end{align*} 
Similar to the online solution, the backward loss function is equivalent to the conditional IB, and we can not obtain the results by standard IB algorithm. As a result, we derive an algorithm that solves the backward path.

\section{Batch Data}\label{sect:batch_results}
In this section, we consider the {\em batch data}, in which $k$ i.i.d. samples are observed at one time. In this setting, we analyze both discrete and continuous distributions, in terms of the fundamental limits and algorithmic method. 
\subsection{Fundamental Limits}
To find $D^{k}_{\text{B}}(R)$, we leverage the equivalence between this distortion-rate problem and the information bottleneck over the Markov chain $X-\theta-X^k-T$. Considering $I(X;T)=H(X)-H(X|T)$ and solving the distortion-rate problem by Lagrange multiplier translates the distortion-rate problem to finding the conditional distribution $p(t|x^k)$ that solves the problem
\begin{equation}\label{eqn:unsupervised.ib}
  \min_{p(t|x^k)} \mathcal{L}=I(X^k;T)-\beta I(X;T),
\end{equation} 
where $I(X^k;T)$ determines the rate, $I(X;T) = H(X) - H(X|T)$ determines the average distortion, and $\beta$ determines the trade-off between the two and dictates which point on the rate-distortion curve the solution will achieve. For {\em discrete} sources, one can use the iterative method proposed in \cite{tishby2000information} to solve (\ref{eqn:unsupervised.ib}). This method only guarantees a {\em local} optimum, but it performs well in practice. When $\theta$ and $X$ are jointly Gaussian, one can use the results in \cite{chechik2005information} for the {\em Gaussian} information bottleneck, in which the optimum $p(t|x^k)$ is Gaussian and given by a noisy linear compression of the source.
\subsubsection*{Discrete Source}
For $k$-sample with a discrete source, when $X$ and $X^k$ are i.i.d. samples drawn from $p(x|\theta)$, the histogram $H_k$ of $X^k$ is a sufficient statistic for $\theta$, and similar to \cite[Theorem.1]{tishby2000information} the optimum mapping $p(t|x^k)$ that minimize (\ref{eqn:unsupervised.ib}) is computed as:
\begin{equation*}
  q(t|H_k) = \frac{q(t)}{Z(H_k,\beta)}\exp\left[-\beta D_{KL}\left[p(x|H_k)||p(x|t)\right]\right],
\end{equation*} 
where $Z(H_k,\beta)$ is the normalization function and $\beta$ is the Lagrangian multiplier in equation \ref{eqn:unsupervised.ib}.

\subsubsection*{Gaussian Source}
When $X$ and $\theta$ are jointly (multivariate) Gaussian, one can appeal to the Gaussian information bottleneck, \cite{chechik2005information}, where it is shown that the optimum $T$ is a noisy linear projection of the source, which one can find iteratively or in closed form and will be described in the sequel. If $X$ is a $d$-dimensional multivariate Gaussian, however, naive application of this approach requires one to find an $dk \times dk$ projection matrix. Here again, exploitation of the structure of this problem simplifies the result.\par
Without loss of generality, let $p(x|\theta) = \mathcal{N}(\theta,\Sigma_x)$ and $p(\theta) = \mathcal{N}(0,\Sigma_\theta)$, where $\Sigma_x \in \mathbb{R}^{d \times d}$ is the covariance of the data, and $\Sigma_\theta$ is the covariance of the prior. For $k$-sample training set $X^k$ with a Gaussian source, the sample mean, $\overline{X}_k=\frac{1}{k}\sum_{i=1}^k x_i$ is a sufficient statistic for $\theta$ and we have the Markov chain $X - \theta - \overline{X}_k - T$, where $X$ is a hypothetical test point conditionaly independent of $X^k$ given $\theta$. \cite{chechik2005information} shows that the IB-optimum compression is $T = A\overline{X}_k + Z$, where $Z$ is white Gaussian noise and  $A$ is a matrix whose rows are scaled left eigenvectors of the matrix
\begin{equation*}
  M = \frac{\Sigma_x+\Sigma_{\theta}}{k}+\frac{k-1}{k}\Sigma_{\theta} - \Sigma_\theta (\Sigma_{x} + \Sigma_\theta)^{-1}\Sigma_\theta.
\end{equation*}
Specifically, let $\lambda_1,\lambda_2,\dots$ be the ascending eigenvalues of $M_k$ and $v_1,v_2$ be the associated left eigenvectors. For $\beta>1$, only the eigenvectors $v_i$ such that $\beta_i := (1-\lambda_i)^{-1} < \beta$ are incorporated into $A$, and the resulting rate-distortion point is given parametrically by
\begin{align*}
  R(\beta) &= \frac{1}{2}\sum_{i=1}^{n(\beta)} \log\left((\beta-1)\frac{1-\lambda_i}{\lambda_i} \right) \\
  D(\beta) &= \frac{1}{2}\sum_{i=1}^{n(\beta)} \log\left(\lambda_i \frac{\beta}{\beta - 1}\right) + H(X),
\end{align*}
where $n(\beta)$ is the number of eigenvalues satisfying $\beta_i < \beta$. For $\beta \approx 1$, the rate is small and the distortion is close to the maximum value $H(X)$; as $\beta \to \infty$ the rate becomes large and the distortion converges to $H(X|X^k)$.

The problem setting implies structure on the compression matrix $A$ beyond what is obvious from above.
\subsection{Algorithm}
\subsubsection*{Discrete Source}
Again, for $k$-sample with a discrete source, when $X$ and $X^k$ are i.i.d. samples drawn from $p(x|\theta)$, the histogram $H_k$ of $X^k$ is a sufficient statistic for $\theta$, and we compute $p(x^k)$, $p(x|X^k)$ and $p(x,x^k)$ in terms of the histogram. Then, the iterative IB algorithm proposed in \cite{tishby2000information}, can be rewritten
\begin{align*}
  q^{(n)}(t|H_k) &= \frac{q^{(n)}(t)}{Z^{(n)}(H_k,\beta)}\exp\left[-\beta D_{KL}\left[p(x|H_k)||p(x|t)\right]\right]\\
  q^{(n+1)}(t) &= \sum_{\theta}q^{(n)}(t|H_k)p(H_k)\\
  q^{(n+1)}(x|t) &= \frac{1}{q^{(n)}(t)}\sum_{\theta}q^{(n)}(t|H_k)p(H_k,x),
\end{align*} 

where $n$ is the iteration index, $q^{(n)}(t|x^k)$ is the choice for $p(t|x^k)$ at iteration $n$, and the other iterated distributions $q^{(n)}(t),q^{(n)}(x|t)$ are intermediate distributions. The number of terms in the distribution $q^{(n)}(t|x^k)$ is upper bounded to $|\mathcal{T}|k^{|\mathcal{X}|-1}$ entries by using the histogram of $X^k$. This distribution must be updated every iteration, and computing the distributions $q^{(n+1)}(t)$ and $q^{(n+1)}(x|t)$ requires summing over $q^{(n)}(t|x^k)$, $p(x^n)$, and $p(x^k,x)$, the latter two of which are computed based on the histogram and have at most $|\mathcal{T}|\cdot k^{|\mathcal{X}| - 1}$, $k^{|\mathcal{X} - 1|}$ entries, respectively. Standard cardinality bounds would suggest that $|\mathcal{T}| = \min\{|\Theta| , k^{|\mathcal{X}| - 1}\}$ is sufficient to achieve $R^{k}_{\text{B}}(D)$. This bound on $\mathcal{T}$ is established along with a fact that $T$ depends on $X^k$ only through $\theta$. 

As a result, the problem structure allows one to reduce the complexity of the IB algorithm from exponential to polynomial when $|\mathcal{X}|$ is constant, albeit of a potentially large degree.
\begin{figure}[htb]
\centering
  \includegraphics[width=\columnwidth]{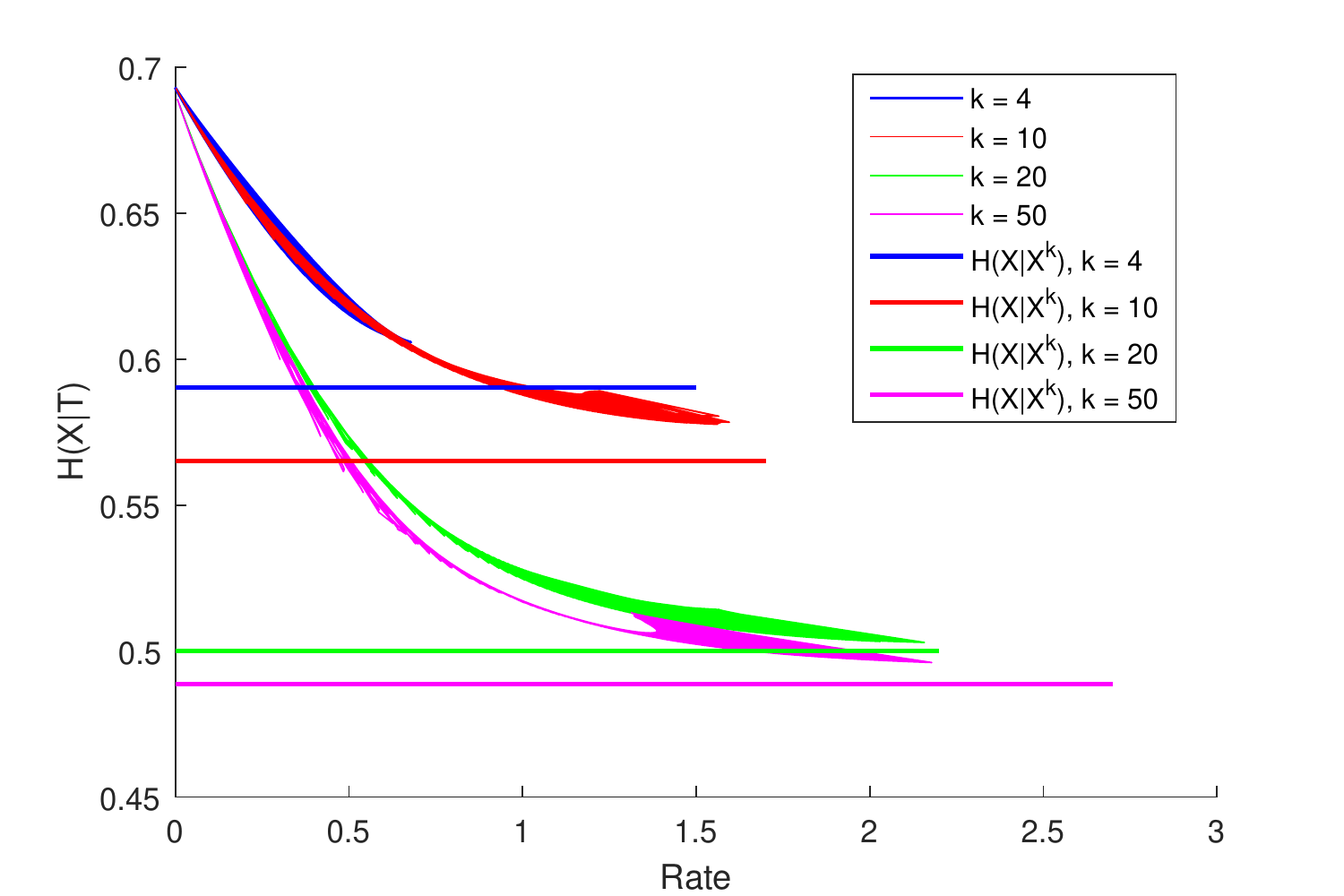}
\caption {Rate-distortion curve for Bernoulli distribution with uniform prior and $k\in\{4,10,20,50\}$.}
\label{fig:Rate-Discrete}
\end{figure}
In Figure~\ref{fig:Rate-Discrete} we plot the rate-distortion curve for a Bernoulli distribution with a uniform prior and $|\mathcal{T}|=k+1$ is sufficient for the Bernoulli distribution. As the number of samples increase we see that the cross entropy loss decreases; however, the cross entropy loss for each sample is bounded by $H(X|X^k)$.
\subsubsection*{Gaussian Source}
For $k$-sample $X^k$ with the Gaussian distribution, the sample mean $\overline{X}_k$ is the sufficeint statistic for $\theta$, and the Markov chain is $X - \theta - \overline{X}_k - T$. \cite{chechik2005information} propose an iterative algorithm to compute the compressed representation $T^{(n)}=A^{(n)}\overline{X}_k+Z^{(n)}$, where the projection matrix $A^{(n)}$ and covariance of noise $Z^{(n)}\sim \mathcal{N}(0,\Sigma_Z^{(n)})$ is computed as:
\begin{align*}
&\Sigma_Z^{(n+1)} = (\beta\Sigma_{t^{(n)}|x}-(\beta-1)\Sigma_{t^{(n)}})^{-1}\\
&A^{(n+1)}=\beta\Sigma_Z^{(n)}\Sigma_{t^{(n)}|x}^{-1}A^{(n)}(I-\Sigma_{x|\overline{X}_k}\Sigma_{\overline{X}_k}^{-1}),
\end{align*}
where $\Sigma_{t^{(n)}|x}$ and $\Sigma_{t^{(n)}}$ are covariance matrix which is computed based on $T^{(n)}$ in each iteration.
\begin{figure}[htb]
\centering
  \includegraphics[width=\columnwidth]{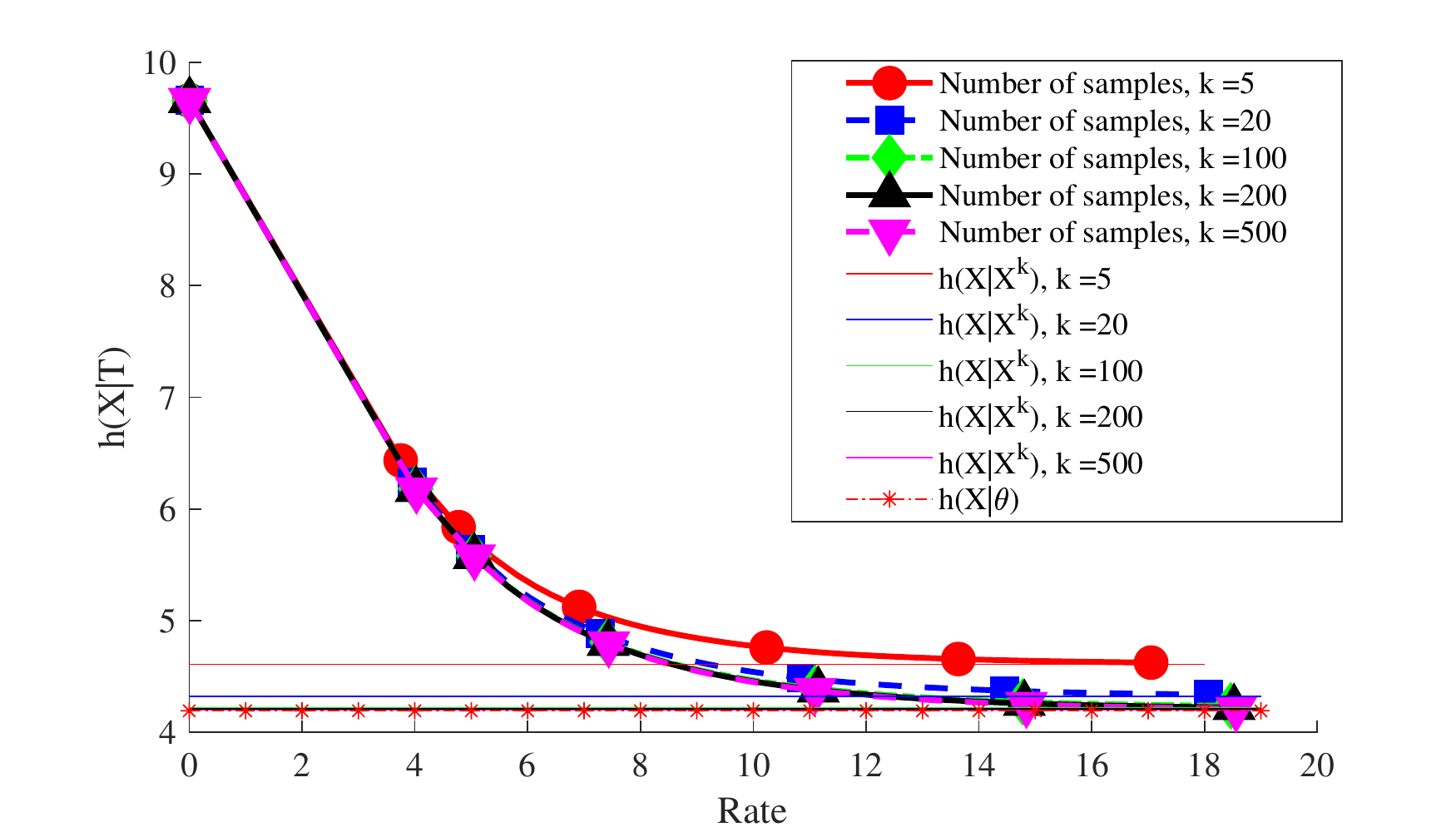}
\caption {Rate-distortion curve for jointly Gaussian distribution with $d=6$ and $k\in\{5,20,100,200,500\}$.}
\label{fig:Rate-Dist}
\end{figure}

In Figure~\ref{fig:Rate-Dist} we plot the rate-distortion curve for a Gaussian source, where $d=6$ and the covariances are drawn elementwise at random from the standard normal distribution and symmetrized. The curve is smooth because the Gaussian information bottleneck provides an exactly optimum solution.

Finally, we show the dependency of the rate on the number of samples $k$ for Gaussian sources. For the Gaussian sources, as $k \to \infty$ the gap between the distortion $h(X|X^k)$ given direct access to the training set and the best case distortion $h(X|\theta)$ goes to zero; equivalently, the mutual information gap $I(X;\theta) - I(X;X^k)$ goes to zero. Figure \ref{fig:Rate-sample} shows the gap between $h(X|T)$ and $h(X|X^k)$ for different rate functions. As it demonstrate for $R = \Omega(\log(k))$ the gap between $h(X|T)$ and $h(X|X^k)$ goes to zero after 6 samples, also this function for rate has the optimum decay of the gap between $h(X|T)$ and $h(X|\theta)$ among other functions for rate.
\begin{figure}[htb]
\centering
  \includegraphics[width=\columnwidth]{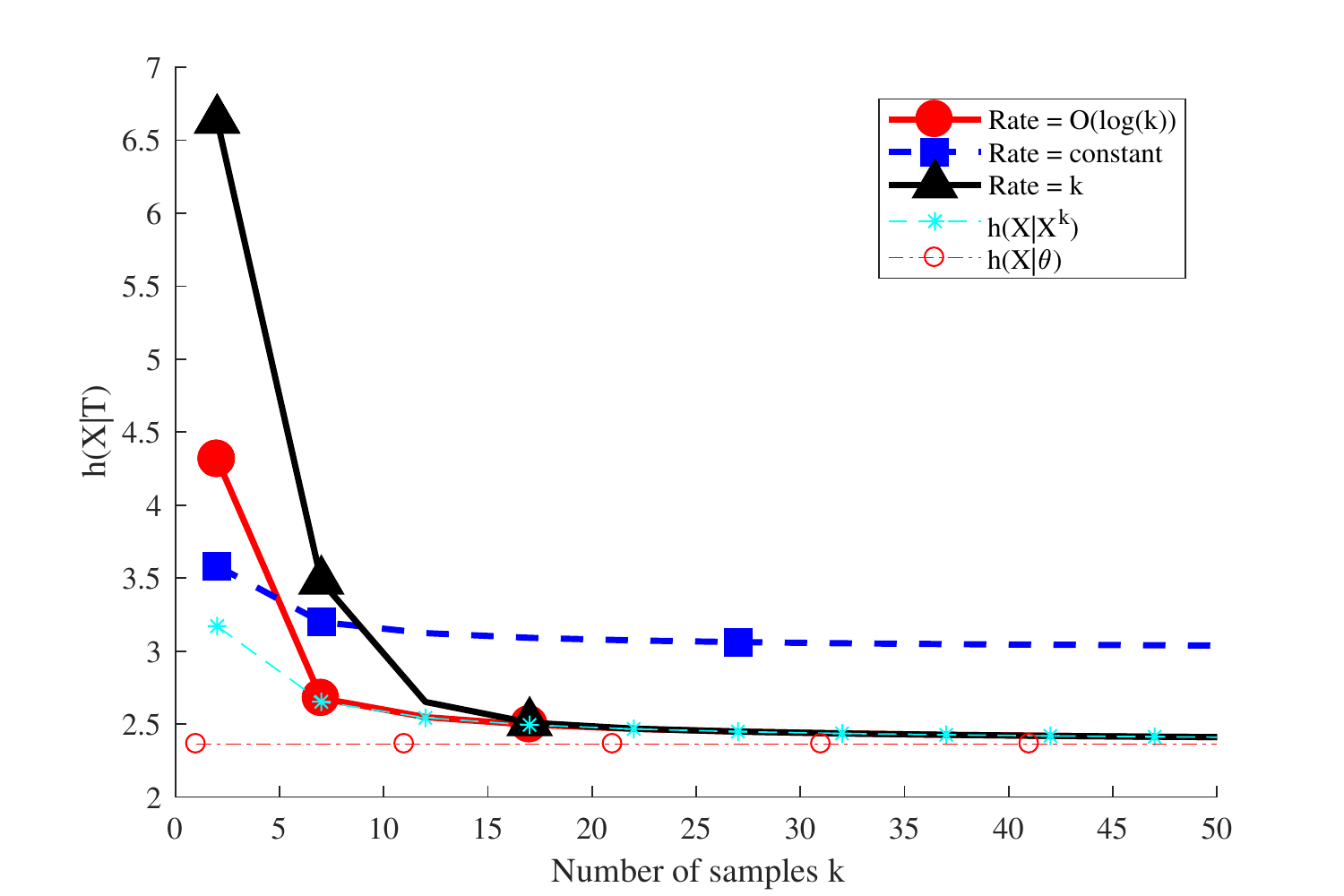}
\caption {Number of Samples-distortion curve for different rate functions against the $h(X|X^k)$ and $h(X|\theta)$.}
\label{fig:Rate-sample}
\end{figure}
\section{Sequential Data}\label{sec:sequential_results}
There are three approaches to solve the distortion-rate problem for  sequential data: in the {\em comprehensive solution} we pose a global problem and find the optimum solution by considering causality for a block data with a known length which is encoded sequentially. In the {\em online solution} for a streaming data,  which is a sequential data with unknown length, we find the optimum compression $T_k$ for each round and consider that the compression of other rounds $T^{-k}$ are constant. Finally, in the {\em two-pass solution}, we add a backward path to the online solution to improve the results; in this setup, similar to the comprehensive solution, we know the length of the block data, and we want to process this data one by one. In the following part, we study {\em fundamental limits} and {\em algorithmic method} for all these three solutions.  

\subsection{Fundamental Limits}
\subsubsection*{Comprehensive Solution}
Let's assume $(X,\theta)$ are jointly Gaussian and we observe a set of $k$ i.i.d. samples $X^k = (X_1,X_2,...,X_k)$ and $X$ is a hypothetical test point conditionally
independent of $X^k$ given $\theta$ (Figure \ref{fig:stream_model}). In sequential data similar to the batch data, sample mean $\overline{X}_k=\frac{1}{k}\sum_{i=1}^k x_i$ is a sufficient statistic for $\theta$, and the Markov chain is $X-\theta-\overline{X}_k-T^k$.
Then, the computation of $D^k_\text{SC}(R)$ is equivalent to the solution of the this problem

\begin{equation}\label{eq:loss_com}
\begin{split}
\max_{{p(t_1|x^1),\dots,p(t_l|x^l)}}\!\!\!\!\!\!\!\!\!\!\!\mathcal{L} = &\sum_{l=1}^k I(X;T^l)-\alpha_l I(\overline{X}_l;T_l|T^{-l})
\end{split}
\end{equation}
where $I(T^l;X)=h(X)-h(X|T^l)$ determines the average distortion of all rounds and $I(S_l;T_l|T^{-l})$ determines the rate in each round, $\alpha_l$ determines the trade-off between the average distortion and rate of each round, and as it is shown in \cite{chechik2005information} $T_l$ is a noisy linear projection of the source. The comprehensive solution finds the optimum compression for each round by solving the global problem since this problem is not a convex problem, finding a closed-form solution is highly unlikely. In addition, finding a numerical solution for a Gaussian distribution and even a discrete distribution is computationally expensive as dimension of data $d$ and the number of samples $k$ become large. This is because we need to compute $k(d\times d)$ elements of projection matrix $A$ for $k$ samples maximizing the (\ref{eq:loss_com}). 
 
\subsubsection*{Online Solution} 
In the online solution for a streaming data, computation of the $D_{SO}^k(R)$ equivalent to a conditional information bottleneck, which solves the problem
\begin{equation}\label{online_loss}
\min_{A_k}\mathcal{L}=I(\overline{X}_k;T_k|T^{-k})-\beta I(X;T_k|T^{-k}),
\end{equation}

where $T^{-k}:=\{T_1,T_2,\dots,T_{k-1}\}$. In this setup $\beta$ determines the trade-off between rate and relevant information and for $\beta \approx 1$, the rate of each round $I(T_k;\overline{X}_k|T^{-k})$ is small and the distortion $H(X|T^k)$ is close to $H(X|T^{-k})$, when $\beta \to \infty$ the rate becomes large and the distortion converges to $H(X|\overline{X}_k)$. In the following theorem, we characterize the optimum $A_k$:
\begin{theorem}\label{theorem:1}
The optimum projection $A_k$ that solves (\ref{online_loss}) for some $\beta$ satisfies
  \begin{equation}\label{projection}
    A_k=
    \begin{cases}
      [0,\dots,0] & 0<\beta<\beta_{c_1} \\
      [\alpha_1v_1^T,0,\dots] & \beta_{c_1}<\beta<\beta_{c_2}\\
             .            \\
             .            \\
             .            \\
      [\alpha_1v_1^T,\dots,\alpha_kv_k^T] & \beta_{c_{k-1}}<\beta<\beta_{c_{k}}
    \end{cases},
  \end{equation}
where $\lambda_1,\lambda_2, . . .$ are the ascending eigenvalues of
$M = \Sigma_{\overline{X}_k|T^{-k},X}\Sigma^{-1}_{\overline{X}_k|T^{-k}}$, $v_1, v_2, \dots$ are the associated left eigenvectors, and $\beta_{c}=\frac{1}{1-\lambda_i}$ are critical values for $\beta$ and $\alpha_i=\sqrt{\frac{\beta(1-\lambda_i)-1}{\lambda_i(v_i^T\Sigma_{\overline{X}_k|T^{-k}}v_i)}}$.
\end{theorem}
\begin{proof}
We first rewrite the loss function in terms of the entropies
\begin{align*}
\min_{A_k}\mathcal{L}&=I(\overline{X}_k;T_k|T^{-k})-\beta I(X;T_k|T^{-k})\\ &=(1-\beta)h(T_k|T^{-k})-h(T_k|\overline{X}_k)+\beta h(T_k|T^{-k},X)\\
 &=\!(1\!-\!\beta)\log(\Sigma_{T_k|T^{-k}}\!)\!-\!\log(\Sigma_Z)\!+\!\beta\log(\Sigma_{T_k|X,T^{-k}}\!)\\
 &=\!(1\!-\!\beta)\!\log(A_k\!\Sigma_{\overline{X}_k|T^{-k}}\!A_k^T\!\!+\!\!I_d)\\
 &+\beta\log(\!A_k\Sigma_{\overline{X}_k\!|\!X,T^{-k}}A_k^T\!\!+\!\!I_d\!),
 \end{align*}
By taking derivative of the loss function with respect to the $A$
\begin{align*}
\frac{\delta\mathcal{L}}{\delta A_k}&=(1-\beta)(A_k\Sigma_{\overline{X}_k|T^{-k}}A_k^T+I_d)^{-1}2A_k\Sigma_{\overline{X}_k|T^{-k}}\\ &+ \beta(A_k\Sigma_{\overline{X}_k|X,T^{-k}}A_k^T+I_d)^{-1}2A_k\Sigma_{\overline{X}_k|X,T^{-k}}
\end{align*} 
In order to obtain minimum of the loss function $\mathcal{L}$, we set $\frac{\delta\mathcal{L}}{\delta A_k}=0$ then we have:
\begin{align}\label{eq:eigeneqution}
&\frac{\beta-1}{\beta}\left[(A_k\Sigma_{\overline{X}_k|T^{-k},X}A_k^T+I_d)(A_k\Sigma_{\overline{X}_k|T^{-k}}A_k^T+I_d)^{-1}\right]A_k\\\nonumber
&=A_k\left[\Sigma_{\overline{X}_k|X,T^{-k}}\Sigma_{\overline{X}_k|T^{-k}}^{-1}\right]
\end{align}
This equation is an eigenvalue problem and $A$ is the eigenvector of $\Sigma_{\overline{X}_k|X,T^{-k}}\Sigma_{\overline{X}_k|T^{-k}}^{-1}$. Then we can substitute $A_k=UV$ and $V\Sigma_{\overline{X}_k|X,T^{-k}}\Sigma_{\overline{X}_k|T^{-k}}^{-1}=LV$ similar to \cite{chechik2005information} and rewrite the (\ref{eq:eigeneqution}):
\begin{align*}
&\frac{\beta\!-\!1}{\beta}\left[(U\Sigma_{\overline{X}_k|T^{-k},X}U^T+I_d)(U\Sigma_{\overline{X}_k|T^{-k}}U^T+I_d)^{-1})\right]U\\
&=UL
\end{align*}
Considering $\Sigma_{\overline{X}_k|T^{-k}}=V^{-1}SV$ and $\Sigma_{\overline{X}_k|T^{-k},X}=V^{-1}SLV$ and multiplying by $U^{-1}$ from left and $U^{-1}(U\Sigma_{\overline{X}_k|T^{-k}}U^T+I_d)^{-1}$ by right we will have:
\begin{align}\label{eqn:Ucomputation}
UU^T=\left[\beta(I-L)-I\right](LS)^{-1}
\end{align}
Therefore, $A_k=UV$, in which $V$ is the eigenvector of the $\Sigma_{\overline{X}_k|T^{-k},X}\Sigma_{\overline{X}_k|T^{-k}}$ and $U$ is computed based on (\ref{eqn:Ucomputation}). 
\end{proof}
\subsubsection*{Two-pass Solution}
To improve the result of the online solution and make the results closer to the optimum solution, we introduce the two-pass version of the online solution that solves the problem from $(k-1)$-th round to the first round by solving following loss function 
\begin{align}\label{eq:recursive online}
\min_{A_{k-1}}\!\mathcal{L}\!=\!I(\overline{X}\!_{k-1};T\!_{k-1}\!|T^{-(k-1)}\!,T_k)\!-\!\beta I(X;T\!_{k-1}\!|\!T^{-(k-1)}\!,\!T_k\!)
\end{align} 
where $T^{-(k-1)}=\{T_1,T_2,\dots,T_{k-2}\}$. In this setup, similar to the online solution, $\beta$ shows the trade-off between rate and relevant information and for $\beta \approx 1$, the rate of each round $I(\overline{X}_{k-1};T_{k-1}|T^{-(k-1)})$ is small and the distortion $H(X|T^k)$ is close to $H(X|T^{-(k-1)},T_k)$, when $\beta \to \infty$ the rate becomes large and the distortion converges to $H(X|\overline{X}_k)$. Similar to Theorem \ref{theorem:1}, we characterize the optimum projection matrix from $(k-1)$-round to the previous round $(k-2)$-th round in the following theorem:
\begin{theorem}\label{theorem2}
The optimum projection matrix $A_{k-1}$ that solves (\ref{eq:recursive online}) 
   \begin{equation}
    A_{k-1}=
    \begin{cases}
      [0,\dots,0] & 0<\beta<\beta_{c_1} \\
      [\alpha_1v_1^T,0,\dots] & \beta_{c_1}<\beta<\beta_{c_2}\\
             .            \\
             .            \\
             .            \\
      [\alpha_1v_1^T,\dots,\alpha_kv_k^T] & \beta_{c_{k-1}}<\beta<\beta_{c_{k}}
    \end{cases},
  \end{equation}
where $\lambda_1,\lambda_2, . . .$ are the ascending eigenvalues of
$K = \Sigma_{\overline{X}_k|T^{-k},X}\Sigma^{-1}_{\overline{X}_k|T^{-k}}$, $v_1, v_2, \dots$ are the associated left eigenvectors, and $\beta_{c}=\frac{1}{1-\lambda_i}$ are critical values for $\beta$ and $\alpha_i=\sqrt{\frac{\beta(1-\lambda_i)-1}{\lambda_i(v_i^T\Sigma_{\overline{X}_k|T^{-k}}v_i)}}$.
\begin{proof}
The proof is similar to that of Theorem \ref{theorem:1} and to avoid repetition we do not write the proof. 
\end{proof}
\end{theorem}
\subsection{Algorithm}
\subsubsection*{Comprehensive Solution}
Comprehensive solution finds the optimum solution for a block of data with a {\it known} length which is encoded sequentially by considering causality. Since we can not find a closed-form solution in the comprehensive case, we solve this optimization problem for $k=2$ using numerical optimization methods for scalar case. Figure \ref{fig:Compare_Stream} demonstrates the total distortion ($H(X|T_1)+H(X|T^2)$) versus the total rate ($I(\overline{X}_1;T_1)+I(\overline{X}_2;T^2|T_1)$), the rate-distortion curve, for both the comprehensive and the online solutions. Although the comprehensive solution converges to the smaller distortion for the same rate, the result of comprehensive solution is scattered. This is because we use the numerical optimization methods (fminunc function in Matlab) to solve (\ref{eq:loss_com}) and in some area this function finds the local minimum. Therefore, we plot the convex hull of the solution since we know the rate-distortion curve is convex.  
 \begin{figure}[htb]
\centering
  \includegraphics[width=\columnwidth]{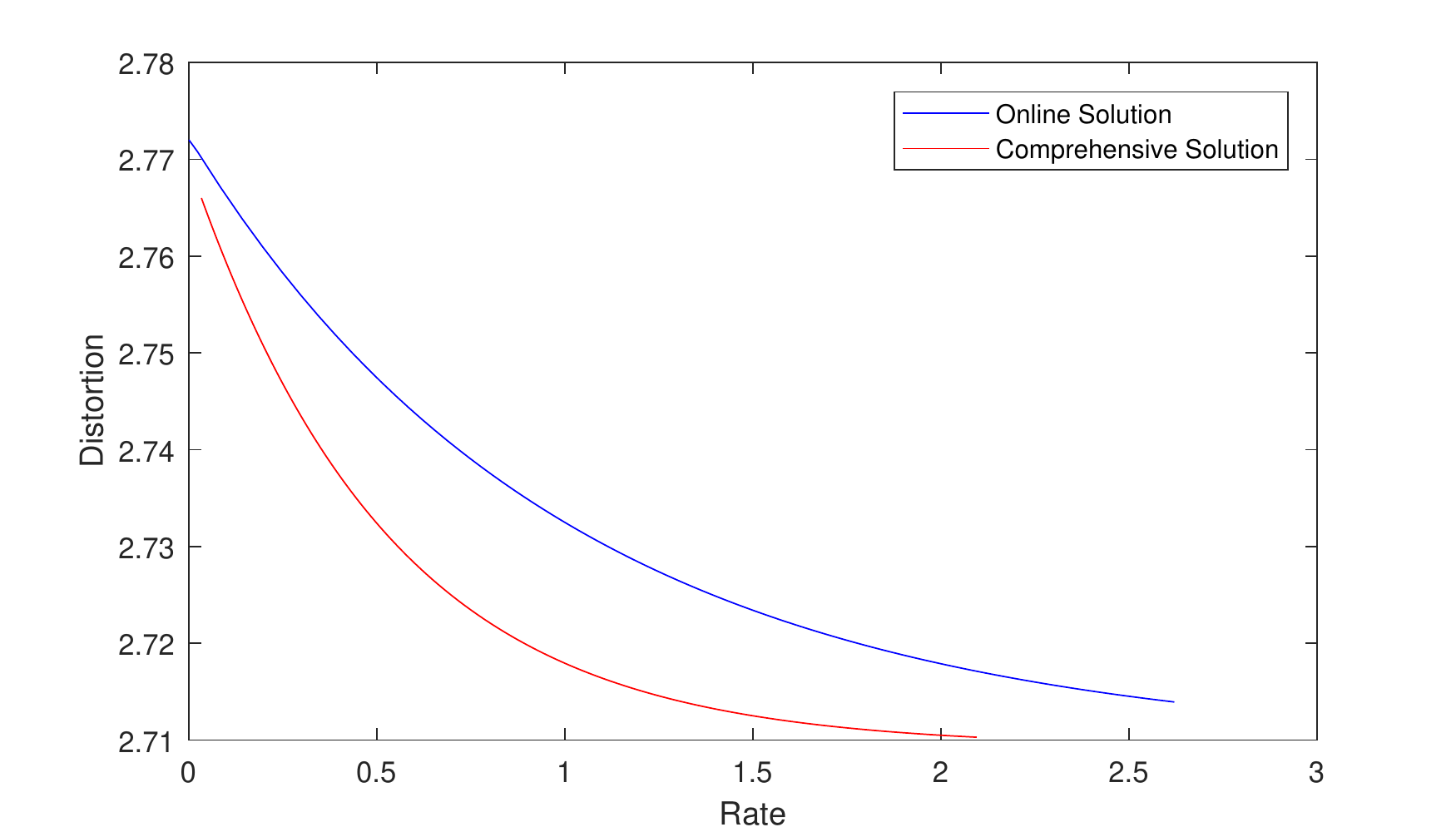}
\caption {Comparison between online and comprehensive solutions by rate-distortion curve, where $K=2$ and "x" axis is the total rates and "y" axis is the total distortion. }
\label{fig:Compare_Stream}
\end{figure}
\subsubsection*{Online Solution} Theorem \ref{theorem:1} states that the optimal projection matrix of each round in the online solution consists of eigenvalues of $M$, in order to compute the whole projection matrix $A$ for the streaming data $X^k$ we propose an iterative algorithm as follows:\par 
\begin{algorithm}[htb]
\caption{Online Solution for Streaming data}
\begin{algorithmic}[1]
\label{algorithm_1}
\STATE {Initiate the values of $\Sigma_{\theta}, \Sigma_n$ and $K$ number of rounds.}
\FOR{each round, $k=1:K$ \do }
\STATE {$\Sigma_{\overline{X}_k}\leftarrow\Sigma_n/n+\Sigma_{\theta}$.}
\STATE {$\Sigma_{T^{-k}}\leftarrow A\Sigma_{X^k}A^T+I_d$.}
\STATE {$\Sigma_C\leftarrow [(\Sigma_x\!+\!(k-1)\Sigma_{\theta})A_1,\dots,(\Sigma_x\!+\!(k-1)\Sigma_{\theta})A_k/k]$}
\STATE {$\Sigma_{\overline{X}_k|T^{-k}}\leftarrow\Sigma_{\overline{X}_k}-\Sigma_C\Sigma_{T^{-k}}^{-1}\Sigma_C^T$}
\STATE {$\Sigma_{\overline{X}_k|X,T^{-k}}\leftarrow\Sigma_{\overline{X}_k}-[\Sigma_{\theta},\Sigma_C]\Sigma_{X,T^{-k}}^{-1}[\Sigma_{\theta},\Sigma_C]^T$.}
\STATE  {Compute $M$ and eigenvalues $\lambda_1,\lambda_2,\dots$ and left eigenvectors $v_1,v_2,\dots$}
\STATE  {Compute the projection matrix of each round, $A_k$ based on \ref{projection}.}
\STATE {Compute the critical values of $\beta$ for each eigenvalue as $\beta_c=\frac{1}{1-\lambda}$, and initiate value for $\beta_{size}$.} 

\STATE {$A=diag(A_1,\dots,A_k)$}
\ENDFOR
\end{algorithmic}
\end{algorithm}
\noindent
where $\beta_{size}$ determines the rank of the projection matrix in each round and depends on the $\beta_c$. In this algorithm, first, for each round (from 1 to $K$) we find the matrix $M$ and its eigenvalues and eigenvectors. Then, we compute the projection matrix $A_k$ of this round according to $\beta_{size}$ and the eigenvalues and eigenvectors of the matrix $M$ and save it to calculate the $M$ for the next round.\par 
Similar to the batch data, from the sufficiency of the sample mean $\overline{X}_k$, we see that regardless of $k$, the optimum compression of each round $T_k$ is a $d$-dimensional representation of the  training set $X^k$; however, in our algorithm the global projection matrix $A$ is a $dk$-dimensional diagonal matrix operator, which we need it to compute the $\Sigma_{T^{-k.}}$. In general, in this problem one can derive the optimum operator from the $d$-dimensional matrix $M$ and compute the projection matrix $A_k$ in each round  and store it into the $A=diag(A_1,\dots,A_k)$. \par
 \begin{figure}[htb]
\centering
  \includegraphics[width=\columnwidth]{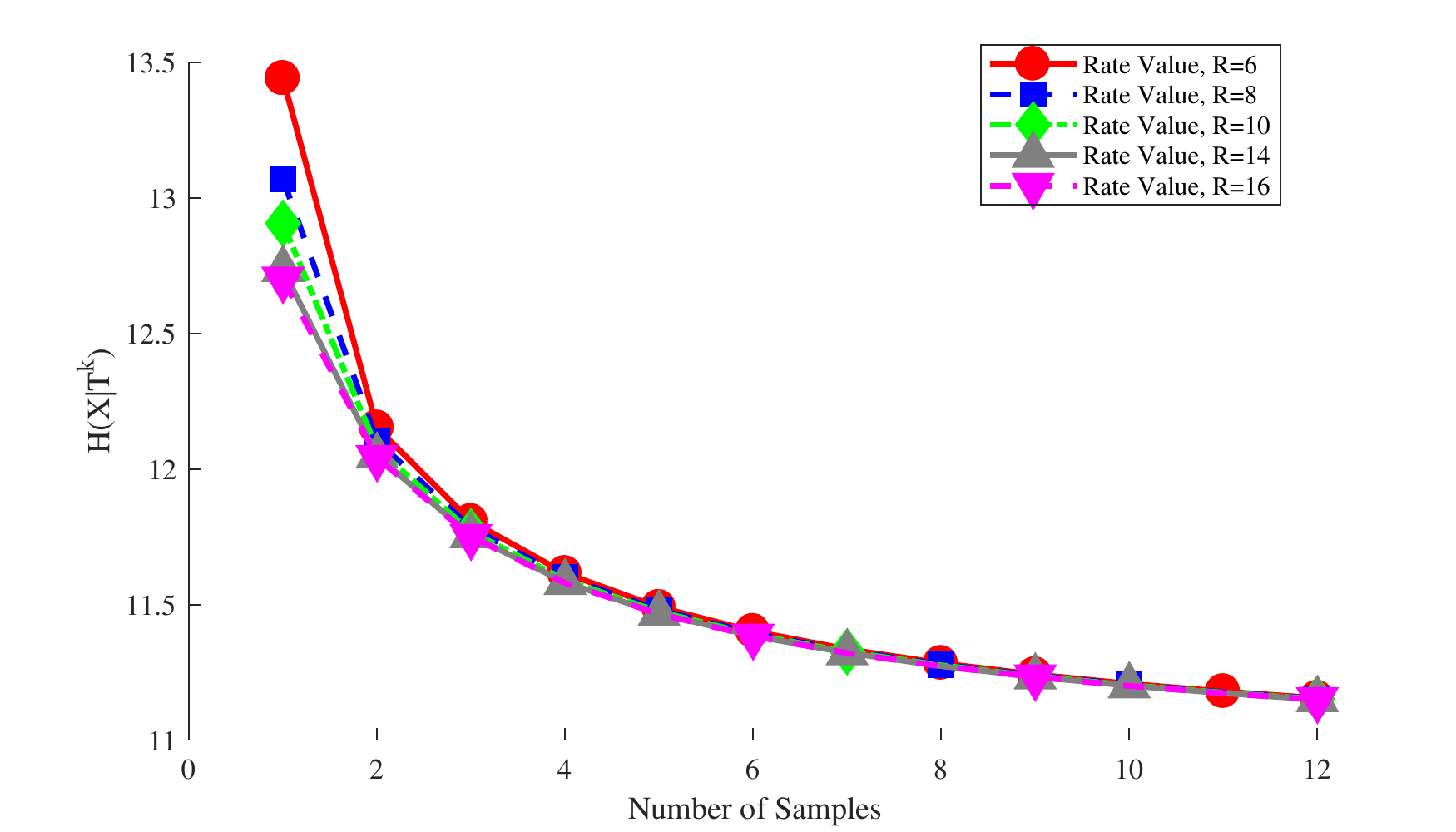}
\caption {Number of Samples-Distortion curve for jointly Gaussian distribution with $d=10$ and $R\in\{4,8,10,14,16\}$.}
\label{fig:Rate-RelStream}
\end{figure}
In Figure \ref{fig:Rate-RelStream} we show the sample-distortion curve for a jointly Gaussian distribution, where $h(X|\theta)$ and $k$ represent the distortion function and the number of rounds or the number of samples that we use, receptively. In addition, for a fixed rate, as the number of samples increase we see that the distortion decreases; however, for a sufficiently large number of samples, this reduction becomes negligible. This can be interpreted as the eigenvalue of the $M$ matrix tends to 1.
\subsubsection*{Two-pass Solution}
In theorem \ref{theorem2} we propose a two-pass solution, where we encode a block data with known length in each round in contrast to the online solution, which we encode the streaming data. Algorithm \ref{algorithm_2} computes the optimum projection of each round (assuming other rounds are constant) based on this chain $A_1-A_2-\dots-A_K$ and then updates the optimum projection of each iteration according to the backward chain $A_K-A_{K-1}-\dots-A_1$. We summarized this procedure in Algorithm \ref{algorithm_2}. 
\begin{algorithm}[htb]
\caption{Two-pass Algorithm}
\begin{algorithmic}[1]
\label{algorithm_2}
\STATE {Initiate the values of $\Sigma_{\theta}, \Sigma_n$, $K$ and $N$ covariance of the $\theta$ and noise, number of rounds and number of iteration that we need to run the two-pass algorithm, respectively.}
\FOR{each round, $n=1:N$ \do }
\FOR{each round, $k=1:K$ \do }
\STATE {$\Sigma_{\overline{X}_k}\leftarrow\Sigma_n/n+\Sigma_{\theta}$.}
\STATE {$\Sigma_{T^{-k}}\leftarrow A\Sigma_{X^k}A^T+I_d$.}
\STATE {$\Sigma_C\leftarrow [(\Sigma_x\!+\!(k-1)\Sigma_{\theta})A_1,\dots,(\Sigma_x\!+\!(k-1)\Sigma_{\theta})A_k/k]$}
\STATE {$\Sigma_{\overline{X}_k|T^{-k}}\leftarrow\Sigma_{\overline{X}_k}-\Sigma_C\Sigma_{T^{-k}}^{-1}\Sigma_C^T$}
\STATE {$\Sigma_{\overline{X}_k|X,T^{-k}}\leftarrow\Sigma_{\overline{X}_k}-[\Sigma_{\theta},\Sigma_C]\Sigma_{X,T^{-k}}^{-1}[\Sigma_{\theta},\Sigma_C]^T$.}
\STATE  {Compute the projection matrix of each round, $A_i$ based on \ref{projection}.}
\STATE {$A=diag(A_1,\dots,A_k)$}
\ENDFOR
\FOR{$k=K-1:1$ \do }
\STATE {$\Sigma_{T^{-k}}\leftarrow$ eliminate $k$-th row and column of the $\Sigma_{T^K}$.}

\STATE {$\Sigma_{\overline{X}_k|T^{-k},T_k}\leftarrow\Sigma_{\overline{X}_k}-\Sigma_C\Sigma_{T^{-k},T_k}^{-1}\Sigma_C^T$}
\STATE {$\Sigma_{\overline{X}_k|X,T^{-k},T_k}\!\!\leftarrow\Sigma_{\overline{X}_k}\!\!-\!\![\Sigma_{\theta},\Sigma_C]\Sigma_{X,T^{-k},T_k}^{-1}[\Sigma_{\theta},\Sigma_C]^T$.}

\STATE  {Compute the projection matrix of each round, $A_k$ based on \ref{projection}.}

\STATE {$A^{(n)}=diag(A_1,\dots,A_k)$}
\STATE {Update $\Sigma_{T^K}\leftarrow A^{(n)}\Sigma_{X^K}A^{(n)T}+I_{Kd}$}
\ENDFOR
\ENDFOR
\end{algorithmic}
\end{algorithm}
In each iteration of this algorithm, at the first part, we compute the covariance matrix $\Sigma_{T^K}$ and projection matrix $A$. In the second part of the algorithm, we compute $\Sigma_{T^{-(k)}}$ based on the $\Sigma_{T^K}$ and then compute the projection matrix of each round $A_k$ and update the whole projection matrix $A^{(n)}$, then update the $\Sigma_{T^K}$ and use it for computation of the next round of the data. At the end, we compute the projection matrix for the $n$-th iteration $A^{(n)}$.\par
 \begin{figure}[htb]
\centering
  \includegraphics[width=\columnwidth]{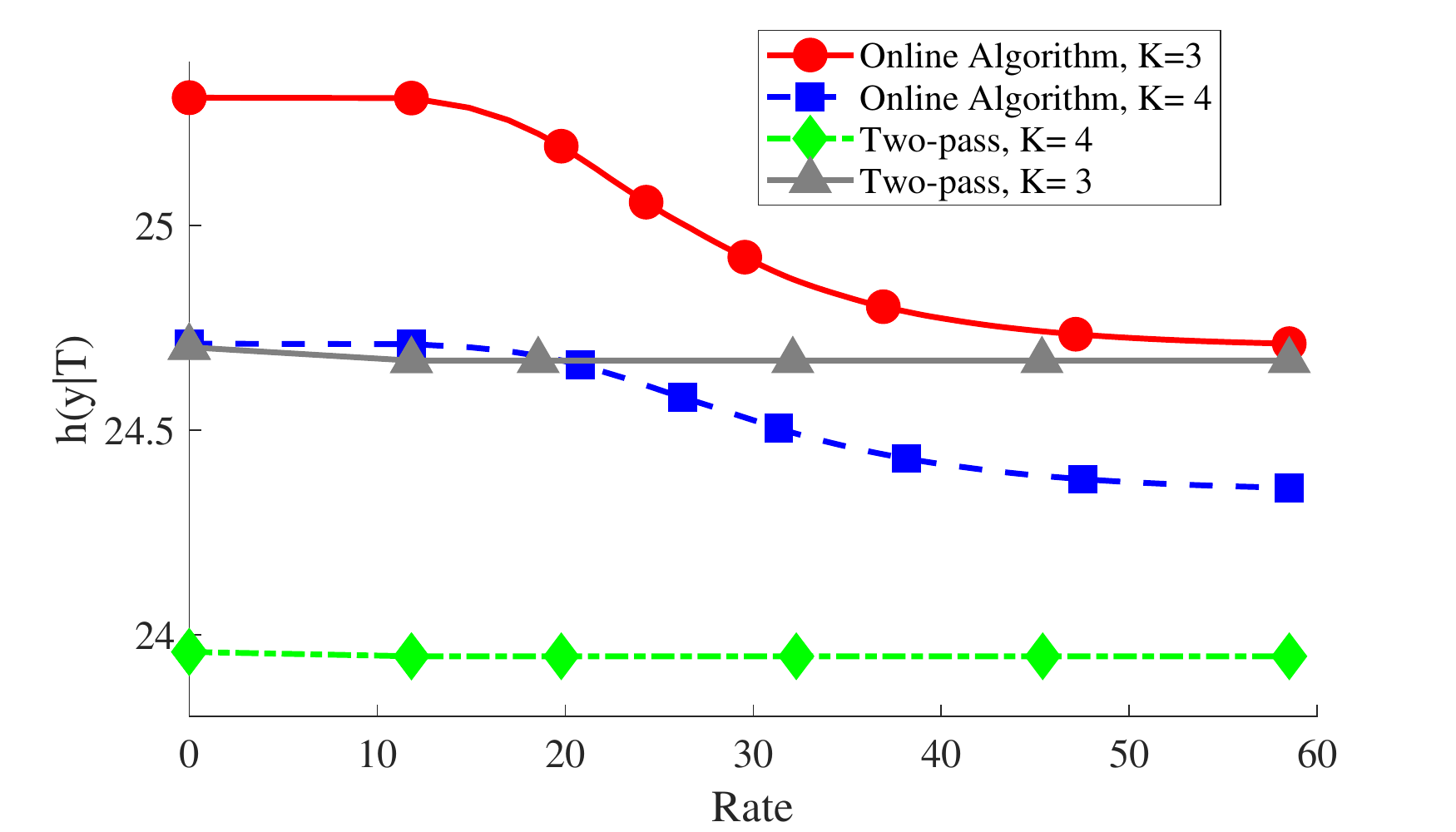}
\caption {Rate-Distortion curve for jointly Gaussian distribution with $d=10$ and $k=\{3,4\}$ for the online algorithm and $k=\{3,4\}$ for two-pass algorithm after one iteration.}
\label{fig:Rate-Dis-Recursive}
\end{figure}
We show the difference between distortion and its lower bound in terms of the rate, in Figure \ref{fig:Rate-Dis-Recursive}, where we show that the two-pass algorithm has the better performance in compare of the online algorithm.
The two-pass algorithm improves the results; however, based on the number of iteration that we want to run the algorithm it needs more computation and time and also it uses the sequential data with the known length, so it is not suitable for the streaming communication.\par
The rate-distortion curve, illustrated in figure \ref{fig:Rate-Dis-Com-Rcr}, compares the comprehensive and two-pass solutions. In this figure similar to the figure \ref{fig:Compare_Stream} we demonstrate the total distortion versus the total rate that we used. Since in the two-pass solution we add a return path to the online solution, this solution converges to the same value of distortion at the cost of higher rate . 

 \begin{figure}[htb]
\centering
  \includegraphics[width=\columnwidth]{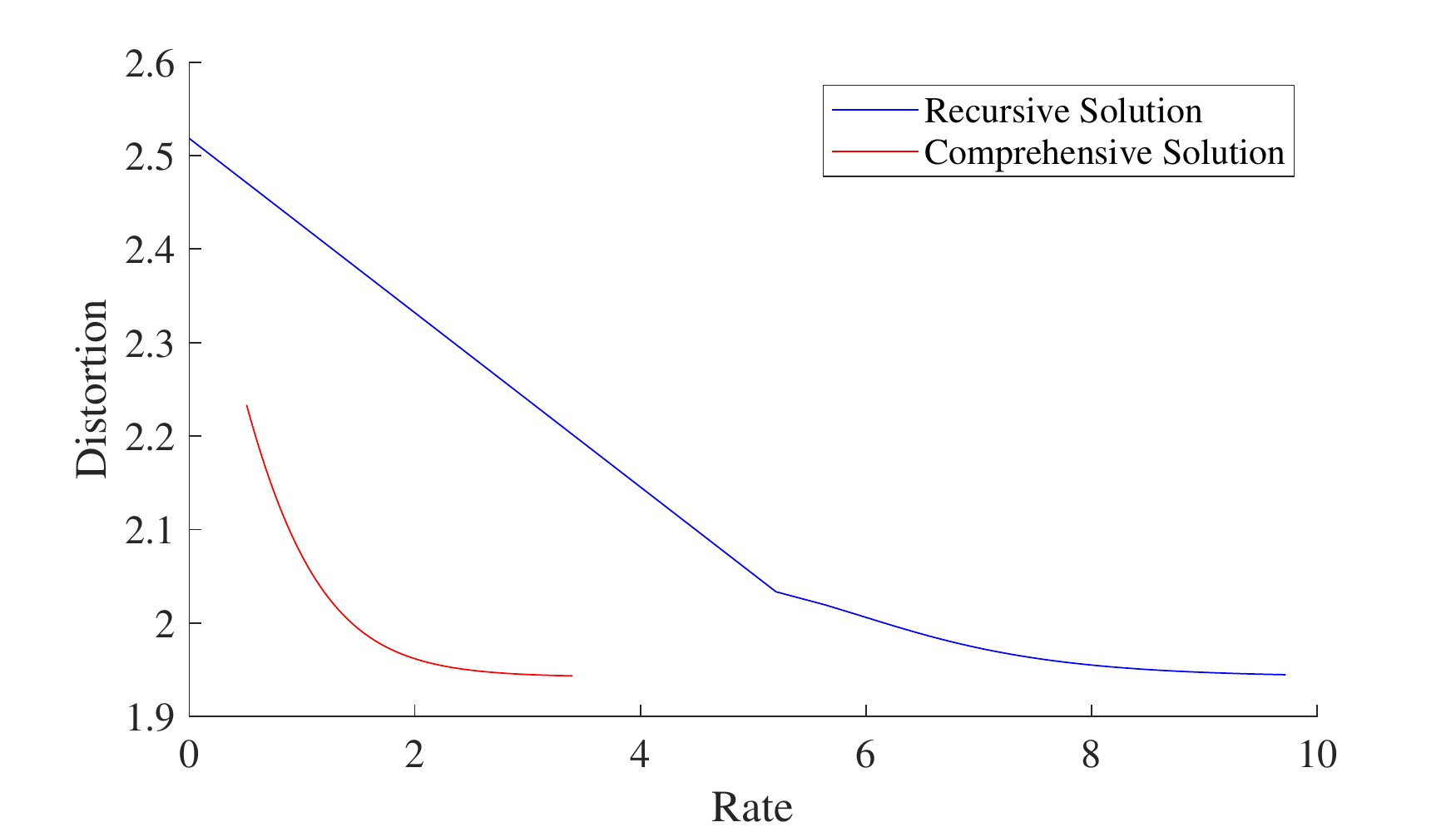}
\caption {Rate-distortion curve for both comprehensive and two-pass solutions, where $K=2$ and x axis is the total rate and y axis is the total distortion. .}
\label{fig:Rate-Dis-Com-Rcr}
\end{figure}

\section{Conclusion}\label{sect:conclusion}
We have examined distributed supervised learning from a {\em compressed} representation of a $k$-sample batch and streaming training set up to a cross-entropy loss, showing that solving for the distortion-rate function is equivalent to an appropriately modified information bottleneck problem.

 We derived a greedy method to solve the distortion-rate function for streaming data as well as an algortithm for this problem. Finally, we improved our results for the streaming data by a new two-pass algorithm. A variety of interesting problems remain to be investigated. The first is supervised learning for batch and streaming dataset. For unsupervised learning, linear compression is sufficient for Gaussian sources per \cite{chechik2005information}; establishing this result for the supervised case, and/or determining the best linear compression for general continuous sources, is a topic for futher study. The second is {\em interactive} compression of geographically separated training sets. For Gaussian unsupervised sources, results in \cite{vera2016collaborative,moraffah2015information} can be used to establish that a single round of interaction is sufficient for optimality; study of the general case is of interest.
 
The last is the study of the effects of single-shot coding on the rate-distortion function as a function of the data distribution and the number of training samples $k$.

\end{document}